\documentclass[conference]{IEEEtran}
\usepackage{epsfig}
\usepackage{amsmath,amssymb}
\usepackage{cite}
\usepackage{amsthm}
\usepackage{subfigure} 

\newtheorem{thm}{Theorem}[]
\newtheorem{lem}{Lemma}[]
\newtheorem{definition}{Definition} []
\newtheorem{rem}{Remark}[]

\begin{document}

\title{{Link-State Based Decode-Forward Schemes for Two-way Relaying}}

\author{\IEEEauthorblockN{Lisa Pinals and Mai Vu}
  \IEEEauthorblockA{Department of Electrical \& Computer Engineering,
    Tufts University, Medford, MA, USA\\
    Email: lisa.pinals@tufts.edu, maivu@ece.tufts.edu}
}

\maketitle

\begin{abstract}
In this paper, we analyze a composite decode-and-forward scheme for the two-way relay channel with a direct link. During transmission, our scheme combines both block Markov coding and an independent coding scheme similar to network coding at the relay. The main contribution of this work is to examine how link state impacts the allocation of power between these two distinct techniques, which in turn governs the necessity of each technique in achieving the largest transmission rate region. We analytically determine the link-state regimes and associated relaying techniques. Our results illustrate an interesting trend: when the user-to-relay link is marginally stronger than the direct link, it is optimal to use only independent coding. In this case, the relay need not use full power. However, for larger user-to-relay link gains, the relay must supplement independent coding with block Markov coding to achieve the largest rate region. These link-state regimes are important for the application of two-way relaying in 5G networks, such as in D2D mode or relay-aided transmission.
\end{abstract}


\IEEEpeerreviewmaketitle
\section{Introduction}\label{sec:intro}
Relay-assisted wireless communication promises to enhance performance in future generation cellular networks. In addition to improving throughput and coverage at the cell edge, relays also offer considerable gain for system capacity. As a result, relays provide a more flexible and efficient use of resources in a dynamic, heterogeneous network. A dedicated relay station and pico base station were proposed to aid communication between base stations and mobiles in 5G networks\cite{dens}. We take this concept one step further and consider a two-way relaying system. Two-way relaying is also utilized for device-to-device (D2D) communications to improve transmission rates and spectral efficiency\cite{d2d}. For example, D2D relaying with either operator or device controlled link establishment is proposed in \cite{d2d5g} in order to realize a two-tiered 5G network. 

Here we consider a two-way relay model, in which a relay, either a dedicated relay node, a base station, or an idle user, helps exchange information between two active users. We assume direct links between users in addition to user-to-relay links, an appropriate model for wireless communication. We focus on the decode-and-forward (DF) relaying strategy as first proposed in \cite{ray2} and as applied in various channel settings \cite{dcfsc,achievable}. When applied to the two-way relay channel (TWRC), there are variations of DF. A DF technique based on random binning \cite{xie2007network} presents an alternative to the original block Markov coding \cite{ref4in2}. A composite DF scheme combining block Markov, independent relaying and partial forwarding strategies has been proposed in \cite{thref}.

As a step toward practical deployment, it is critical to understand the conditions under which a particular scheme outperforms others. The authors of \cite{asymmetric} and \cite{cooperative} have started to analyze hybrid relaying with various constraints. In \cite{asymmetric}, both DF with joint modulation and DF with network-superposition coding are independently studied in the sum rate maximization problem, assuming an asymmetric half-duplex channel model. Under a QoS constraint, the optimal resource allocation in terms of time and power at the relay is derived. The authors of \cite{cooperative} also combine DF and network coding,  yielding a cooperative network coding scheme. Although it is mentioned that the optimal resource allocation depends on the link state, this point is not thoroughly explored. 

In this paper, we specifically examine how the link state impacts the optimal allocation of power in a composite DF scheme \cite{thref} which combines block Markov superposition coding \cite{ray2} with independent coding based on random binning \cite{xie2007network} in a full-duplex TWRC. We are interested in determining which relay technique (or combination of both) yields the largest rate region for a particular set of link states. The weighted sum rate is maximized for a given set of link states and the resulting optimal power allocation parameters ultimately govern the optimal composition of the signal. By optimizing with regard to the link states, our composite scheme achieves a strictly better rate region than existing DF techniques or simple timesharing between these existing schemes. This link-state perspective is useful in 5G systems to dynamically adapt transmission to changing links. Further, it is demonstrated that in some cases the relay does not need to use full power. This aligns with the projected theme of green communications in 5G \cite{fullduplex}. A low power node that increases both capacity and coverage is of critical importance for 5G networks.

\section{Channel Model and Transmission Scheme}\label{sec:system_model}
\subsection{Channel model}
In this paper, we consider a full-duplex two-way relay channel (TWRC). To improve spectral efficiency as well as energy efficiency, full-duplex is advocated for in 5G networks, especially with recent advances in self-interference cancellation \cite{fullduplex}. The full-duplex TWRC can be modeled as
\noindent
\begin{align}\label{gaumod}
Y_1&=g_{12}X_2+g_{1r}X_r+Z_1,\nonumber\\
Y_2&=g_{21}X_1+g_{2r}X_r+Z_2,\nonumber\\
Y_r&=g_{r1}X_1+g_{r2}X_2+Z_r,
\end{align}
\begin{figure}[t] \label{fig1}
    \begin{center}
    \subfigure[]{
		\includegraphics[width=0.22\textwidth]{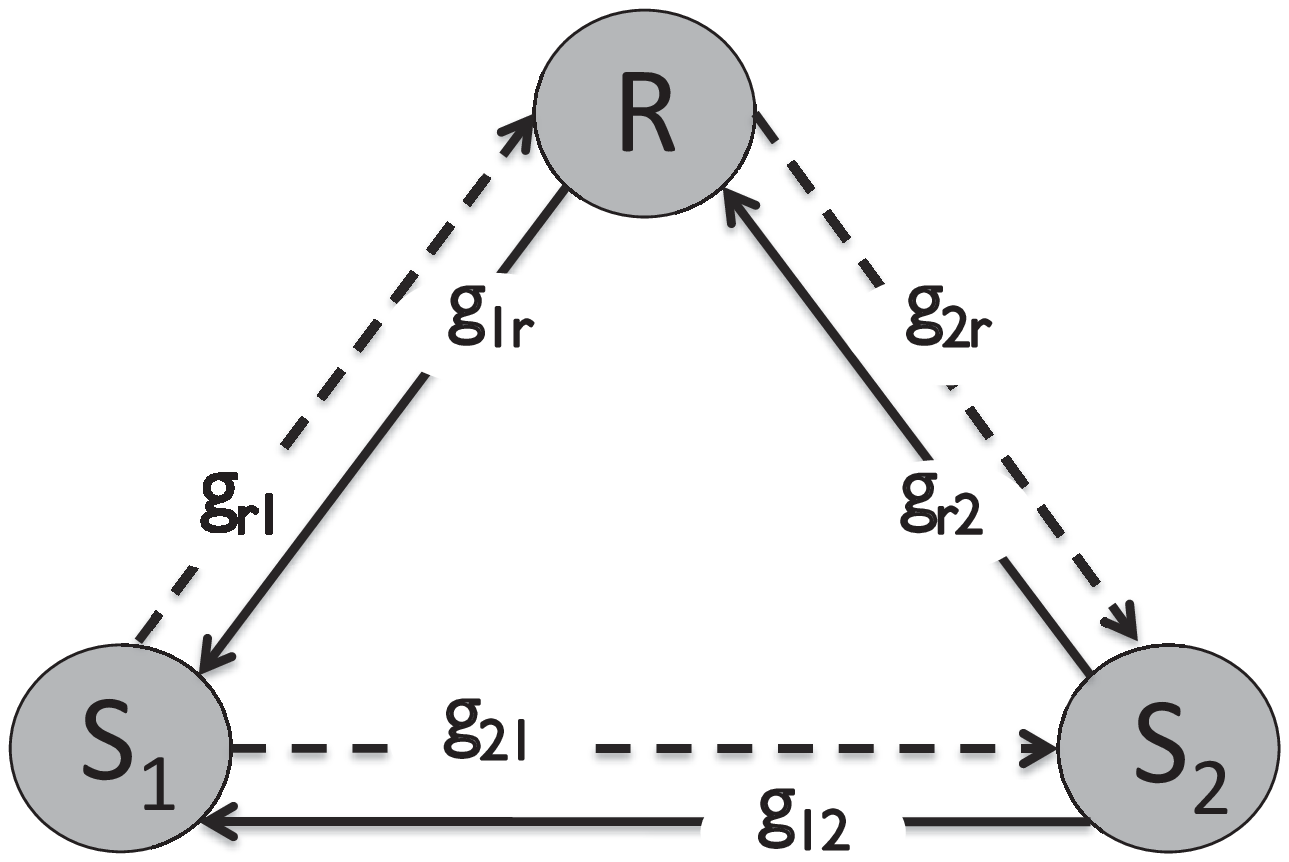}
		\label{fig:channelconfig}
    }
    \subfigure[]{
		\includegraphics[width=0.23\textwidth]{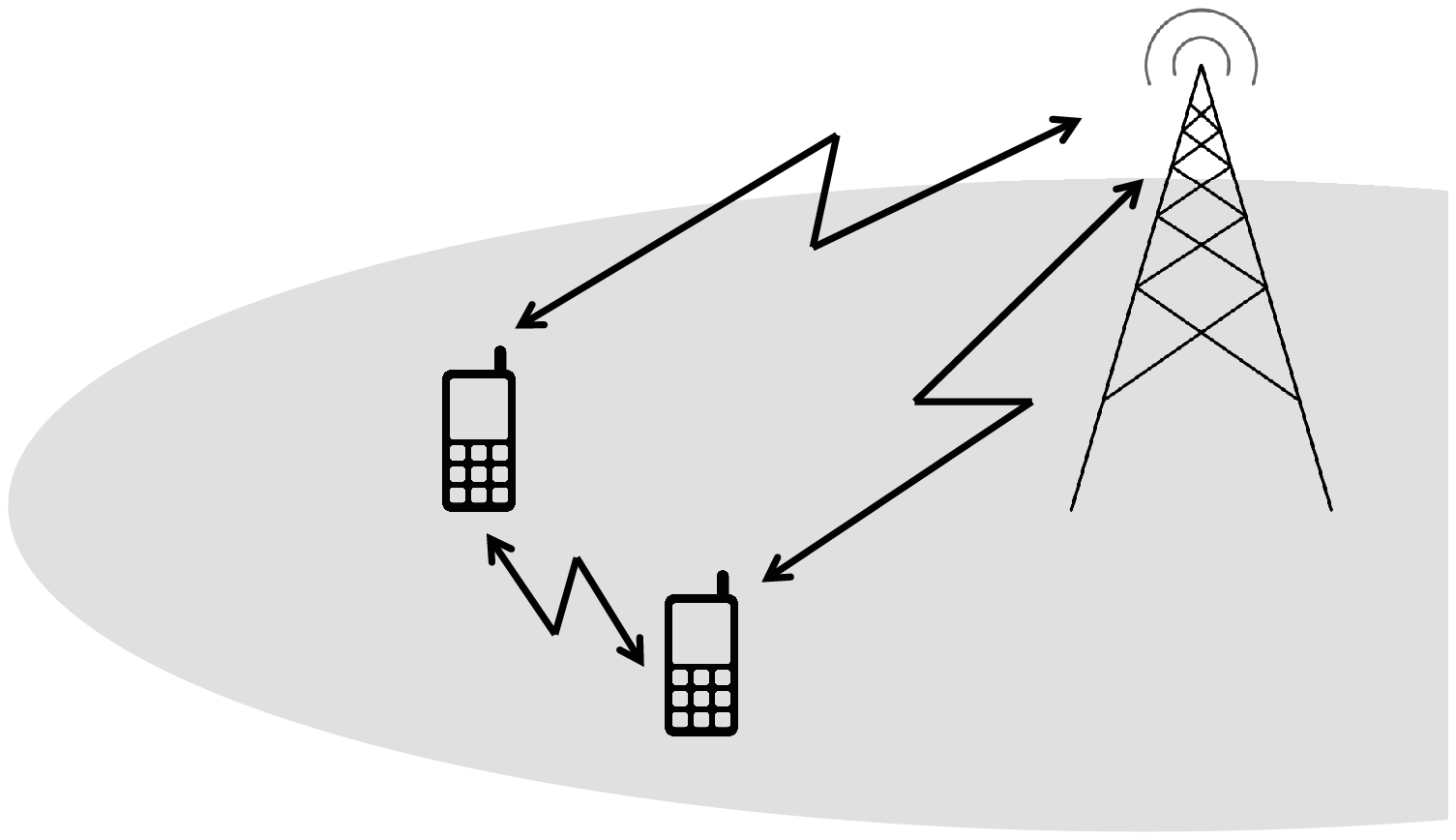}
		\label{fig:cellphone}
    }
    \caption{Two-way relay channel model}
    \end{center}
\vspace*{-7mm}
\end{figure}
\noindent where $Z_{1},Z_{2},Z_{r}\sim$ $\cal{CN}$(0,1) are independent complex Gaussian noises with zero mean and unit variance; ($X_{1}$,$Y_{1}$),  ($X_{2}$,$Y_{2}$), ($X_{r}$,$Y_{r}$) are pairs of the transmitted and received signal at user 1, 2, and the relay, respectively. Without loss of generality, the average input power constraints for all users and the relay can be assumed to be equal to P.

The link gain coefficients are typically complex valued. However, similar to \cite{csi}, we assume the phases of these link coefficients are known at the respective transmitters so that coherent transmission is possible, and the full link coefficients are known at the corresponding receivers. Future work will investigate the relaxation of this assumption in fading channels. Channel state information (CSI) at the transmitter could be obtained in 5G by feeding it back from the respective receiver. As such, the achievable rate depends only on the $amplitude$ of the link coefficients, denoted by $g_{ij}$ (from node $j$ to $i$). This link model is illustrated in Fig. \ref{fig:channelconfig}. Note that this model is the most general version of the TWRC, in which links are not assumed to be reciprocal and the direct links are present. A multi-hop TWRC model without the direct links can be obtained simply by setting $g_{21}$ and $g_{12}$ equal to zero in (\ref{gaumod}). 
 
This general TWRC model could readily be applied in several 5G scenarios. All nodes could be user equipment (UEs), such as in device relaying with device controlled link establishment \cite{d2d5g}. Alternatively, the base station could assist active UEs in D2D mode, as illustrated in Fig. 1(b). Another example is a low powered base station as in \cite{dens} which relays information from a user to an anchor base station with wired backhaul. We extend these examples to the two-way relaying case for improved spectral efficiency. Our composite scheme can be employed to improve transmission rates in any of these scenarios. A more flexible, spectrally efficient 5G network would incorporate all of these scenarios: idle users would have the capability and economic incentives to act as relays, and base stations would have intelligent relaying capabilities in areas demanding increased capacity.


\subsection{Transmission scheme}
We first describe a full-duplex scheme designed for the TWRC based on DF relaying as proposed in \cite{thref}. The relay has the option of transmitting using three distinctive techniques: independent coding, a signal structure that enables block Markov coding at the sources, or a combination of these two approaches. Both sources may perform block Markov coding or independent coding depending on the signal structure at the relay. Here we describe the composite scheme consisting of all of these techniques, then proceed to analyze it in the next section.

\subsubsection{Transmit signal design}
Let's denote the new message user 1 and user 2 send in block $i$ as $m_{1,i}$ and $m_{2,i}$ respectively. The relay partitions the set of all messages of the previous block $\{m_{1,i-1},m_{2,i-1}\}$ equally and uniformly into a number of bins and indexes these bins by $l_{i-1}$. The users and the relay then construct the transmit signals in block $i$ as follows:
\noindent
\begin{align}\label{txsignals}
&X_1\!=\sqrt{\alpha_1}W_1(m_{1,i-1})+\sqrt{\beta_1}U_1(m_{1,i})\\
&X_2\!=\sqrt{\alpha_2}W_2(m_{2,i-1})+\sqrt{\beta_2}U_2(m_{2,i})\nonumber\\
&X_r\!\!=\!\!\sqrt{k_1\alpha_1}W_1(m_{1,i-1})\!+\!\!\sqrt{k_2\alpha_2}W_2(m_{2,i-1})\!+\!\!\sqrt{\beta_3}U_r(l_{i-1})\nonumber
\end{align}
\noindent
where $W_1$, $W_2$, $U_1$, $U_2$, $U_r$ are independent Gaussian signals with zero mean and unit variance that encode the respective messages and bin index. Power allocation factors $\alpha_1$, $\alpha_2$, $\beta_1$, $\beta_2$, $\beta_3$ are non-negative and satisfy
\noindent
\begin{align}\label{powerconstraints}
&\alpha_1+\beta_1 \le P, \ \  \alpha_2+\beta_2 \le P, \  \ k_1\alpha_1+k_2\alpha_2+\beta_3 \le P
\end{align}
\noindent
where $k_1$, $k_2$ are scaling factors that relate the power allocated to transmit the same message at each source and relay in the block Markov signal structure.

Therefore, in each transmission block when $\alpha_i \neq 0, i\in\{1,2\}$, both users send not only the new messages of that block, but also repeat the message of the previous block. This retransmission due to block Markov coding creates a coherency between the signal transmitted from each source and the relay, which ultimately results in a beamforming gain. However, the relay must split its power between $W_1$ and $W_2$ for the retransmission, each of which carries only a single message ($m_{1,i-1}$ or $m_{2,i-1}$). In addition to the block Markovity functionality, the relay also creates a new signal $U_r$ that independently encodes both messages via binning. Using independent coding, one bin index is able to solely represent a message pair $(m_{1,i-1},m_{2,i-1})$.
\subsubsection{Decoding}
At the relay, decoding is quite simple and is similar to the multiple access channel (MAC). The received signal in each block at the relay is
\noindent
\begin{align}\label{Yr}
Y_r\!&=\!g_{r1}(\sqrt{\alpha_1}W_1\!\!+\!\!\sqrt{\beta_1}U_1)\!+\!g_{r2}(\sqrt{\alpha_2}W_2\!\!+\!\!\sqrt{\beta_2}U_2)\!\!+\!Z_r
\end{align}
\noindent
In block $i$, the relay already knows signals $W_1$, $W_2$ (which carry $m_{1,i-1}$, $m_{2,i-1}$), and is interested in decoding $U_1$ and $U_2$ (which carry $m_{1,i}$, $m_{2,i}$). The optimal maximum aposteriori probability (MAP) decoding rule for the relay is
\noindent
\begin{align}\label{decodingruleR}
(\tilde{m}_{1,i},\tilde{m}_{2,i}) &= \operatorname*{arg\,max} P(U_{1,i},U_{2,i}|Y_{r,i},W_{1,i},W_{2,i}).
\end{align}

Sliding window decoding is performed at each user based on signals received in two consecutive blocks. To decode new information sent in block $i$, a user looks at received signals in both blocks $i$ and $i+1$, resulting in a one-block decoding delay. The received signal in each block for user 2 is
\noindent
\begin{align}\label{Y2}
Y_{2} &= g_{2r}(\sqrt{k_1\alpha_1}W_{1}+\sqrt{k_2\alpha_2}W_{2}+\sqrt{\beta_3}U_{r})\nonumber\\
&+g_{21}(\sqrt{\alpha_1}W_{1}+\sqrt{\beta_1}U_{1})+Z_2.
\end{align}
\noindent
Assuming that user 2 has correctly decoded $m_{1,i-1}$, then in block $i$, user 2 knows $W_{1}$, $W_{2}$, and $U_r$ and can subtract them from $Y_{2,i}$. Next, in block $i+1$, user 2 only knows $W_2$, and can directly subtract it from $Y_{2,i+1}$. 

We write this sliding window joint decoding simultaneously over two blocks using the optimal maximum aposteriori probability decoding rule at user 2 as follows:
\noindent
\begin{align}\label{decodingrule2}
\hat{m}_{1,i} &= \operatorname*{arg\,max} P(U_{1,i}|Y_{2,i},W_{1,i},W_{2,i},U_{r,i})\nonumber\\&\cdot P(W_{1,i+1},U_{r,i+1}|Y_{2,i+1},W_{2,i+1},m_{2,i}). 
\end{align}

\noindent

\subsubsection{Achievable Rate Region}
With the above transmit signals and decoding rules, we obtain the following Theorem \ref{th_acheivablerate}.

\begin{thm}\label{th_acheivablerate}
Using the proposed DF based scheme, the following rate region is achievable for the Gaussian TWRC:
\begin{align}\label{RateConstraints}
R_1\!\leq\!&\min\{J_1,J_2\},\ R_2\!\leq\!\min\{J_3,J_4\},\ R_1\!+\!R_2\!\leq\!J_5\\
\text{where }&J_1=C(g_{r1}^{2}\beta_1), \ \ \  J_3=C(g_{r2}^{2}\beta_2)\nonumber\\
&J_2=C(g_{21}^{2}P+2g_{21}g_{2r}\sqrt{k_1}\alpha_1+g_{2r}^{2}k_1\alpha_1+g_{2r}^{2}\beta_3)\nonumber\\
&J_4=C(g_{12}^{2}P+2g_{12}g_{1r}\sqrt{k_2}\alpha_2+g_{1r}^{2}k_2\alpha_2+g_{1r}^{2}\beta_3)\nonumber\\
&J_5=C(g_{r1}^{2}\beta_1+g_{r2}^2\beta_2)
\label{RateConstraints2}
\end{align}
with $C(x)=log(x+1),\ g_{*}$ as amplitudes of link coefficients, and power allocation factors satisfying (\ref{powerconstraints}).
\end{thm}
\begin{proof}
For full proof, see \cite{thref}.
\end{proof}
Intuitively, rate constraints $J_1$, $J_3$, and $J_5$ come directly from decoding the messages at the relay as in (\ref{decodingruleR}), the same as in a MAC. $J_2$ comes from decoding the message at user 2 as described in (\ref{decodingrule2}), in which we leverage the simultaneous decoding over two blocks to achieve a higher rate than separate decoding. $J_4$ is derived similarly to $J_2$, but for user 1 instead of user 2. As such, the transmission rate of each user is constrained by the smaller of the two rates achievable by decoding either at the receiving user or at the relay. 

\section{Rate Region Analysis}\label{sec:analysis}
\subsection{Problem formulation and approach}
In the proposed scheme, the relay has the option to perform independent coding, block Markov coding, or a combination of the two. Sources can transmit with or without block Markov coding. The goal of this section is to investigate 
how link states affect which transmission strategy should be chosen to produce the largest rate region. For a given set of link states, our analysis reveals which coding strategy will produce the largest rate region.


From the rate constraints in (\ref{RateConstraints}), and for some $\mu \in [0,1]$ an optimization problem to maximize the rate region boundary is
\noindent
\begin{align}\label{optimization}
\max \ &\mu R_1+(1-\mu)R_2\nonumber\\
\text{subject} \text{ to }&R_1\! \leq\! \min\{J_1,J_2\},\ R_2\! \leq \! \min\{J_3,J_4\},\ R_1\!+\!R_2\! \leq \! J_5,\nonumber\\
&k_1\alpha_1+k_2\alpha_2+\beta_3\! \le \! P,\ \alpha_1+\beta_1\! \le \! P,\ \alpha_2+\beta_2 \! \le \! P\nonumber\\
&\alpha_1,\alpha_2,\beta_1,\beta_2,\beta_3,k_1,k_2 \geq 0.
\end{align}

When the power allocation is varied, a different pentagonal (or degenerate rectangular) rate region of (\ref{RateConstraints}) is obtained. Typically, only the outermost points from any rate region generated from a particular set of power allocation parameters will be on the boundary of the overall rate region for a given set of link states. This could be either of the two corner points of the rate region. The optimization of these two corner points can be achieved via optimization problem (\ref{optimization}), where $\mu \in [0,1]$. If $\mu \in (1/2,1]$, the rate of user 1 will be prioritized over that of user 2, and the lower corner point will be the optimal point which appears on the boundary of the overall rate region and vice versa for $\mu \in [0,1/2)$. By varying $\mu$ from $0$ to $1$, we trace out the whole boundary of the rate region.

For each value of $\mu$, optimization problem (\ref{optimization}) will give at least one point on the boundary of the overall rate region which corresponds to a particular optimal power allocation for that $\mu$. Since these two cases are reciprocal, we focus on the region where $\mu \in (1/2,1]$.

Optimization problem (\ref{optimization}) is convex and the optimal power allocation parameters can be fully deduced from the KKT conditions and the rate constraints for a given set of link gains. Our goal is to analyze the KKT conditions, combined with the geometry of the rate constraints, to deduce all cases of link-state conditions under which a single (or both) relay technique is required to achieve the largest rate region. Specifically, we look at the various dual variables associated with the rate and power constraints and analyze them for the tightness of those constraints. 
From these dual parameters we deduce the optimal power parameters $\alpha_i,\beta_j, k_i,\ i\in\{1,2\},\ j\in\{1,2,3\}$ and conditions for the link states. The solution will describe which strategy the relay and users should employ in order to compose the optimal transmit signal based on the link states.

\subsection{Link-state based optimal transmission}
Here we state our main results: the analytical link regimes and associated optimal transmission scheme for each regime. We first define these regimes for the user-to-relay link states:

\begin{definition}\label{rateregiondef}
For the user-to-relay link gain for user 2, $g_{r2}$, and in the case that $g_{12}^2(1+g_{r1}^2P)\leq g_{12}^2+g_{1r}^2$ we define link-state regimes as follows:
\begin{align}
&\mathcal{T}1: g_{r2}^2 \in[0,g_{12}^2]\nonumber\\
&\mathcal{T}2: g_{r2}^2 \in(g_{12}^2, g_{12}^2(1+g_{r1}^2P)]\nonumber\\
&\mathcal{T}3: g_{r2}^2 \in(g_{12}^2(1+ g_{r1}^2P),g_{12}^2+g_{1r}^2]\nonumber\\
&\mathcal{T}4: g_{r2}^2 \in(g_{12}^2+g_{1r}^2,(g_{r1}^2P+1)(g_{12}^2+g_{1r}^2)]\nonumber\\
&\mathcal{T}5: g_{r2}^2 \in((g_{r1}^2P+1)(g_{12}^2+g_{1r}^2),\infty)
\end{align}
For the link from user 1 to the relay, $g_{r1}$, we define:
\begin{align}\label{gr1}
&\mathcal{R}1: g_{r1}^2 \in[0,g_{21}^2]\nonumber\\
&\mathcal{R}2: g_{r1}^2 \in(g_{21}^2, g_{21}^2+g_{2r}^2]\nonumber\\
&\mathcal{R}3: g_{r1}^2 \in(g_{21}^2+g_{2r}^2,\infty)
\end{align}
\end{definition}

\renewcommand{\arraystretch}{1.3}
\begin{table}\label{ratetable}
\begin{center}
    \begin{tabular}{l@{\hskip 0.05in} | l@{\hskip 0.05in} | l@{\hskip 0.05in} | l@{\hskip 0.05in} | l@{\hskip 0.05in} | l@{\hskip 0.05in}}
    $g_{r1}^2 \downarrow g_{r2}^2 \rightarrow$ & $\mathcal{T}1$ & $\mathcal{T}2$ & $\mathcal{T}3$ & $\mathcal{T}4$  & $\mathcal{T}5$\\ \hline
         & U1: DT  & U1: DT & U1: DT & U1: DT & U1: DT\\
             $\mathcal{R}1$& U2: DT & U2: Ind & U2: Ind & U2: BM & U2: BM\\ \hline
              & U1: Ind & U1: Ind & U1: Ind & U1: Ind & U1: Ind \\
    $\mathcal{R}2$ & U2: DT & U2: DT & U2: Ind & U2: Ind  & U2: BM\\ \hline
            & U1: BM & U1: BM & U1: BM & U1: Both & U1: Both\\
    $\mathcal{R}3$ & U2: DT & U2: DT & U2: Ind & U2: Both & U2: Both \\ \hline
    \end{tabular}
        \caption{Optimal Technique By Link Regime, $\mu \in (1/2,1$]}
        Where U1 is user 1, U2 is user 2, 'Ind': independent coding only, 'BM': block Markov coding only, 'DT': direct transmission, 'Both': BM and Ind
        \end{center}
\end{table}
\renewcommand{\arraystretch}{1}

\begin{thm}\label{tablethm}
For $1/2<\mu\le 1$ and $g_{12}^2(1+ g_{r1}^2P)\leq g_{12}^2+g_{1r}^2$, Table 1 shows the optimal techniques for both users in a given link-state regime as defined in Definition \ref{rateregiondef}.
\end{thm}
\begin{proof}
See Appendix A. 
\end{proof}

\begin{rem}
The case in which $g_{12}^2(1+ g_{r1}^2P)>g_{12}^2+g_{1r}^2$ is similar but omitted due to space and will be included in an upcoming journal article.
\end{rem}

\begin{rem}
For $0\leq \mu < \frac{1}{2}$, we obtain a transposed table with the coding for user 1 and user 2 reversed.
\end{rem}


\noindent From Theorem 2 and the transmit signal design, it follows that if either user performs independent coding, then $\beta_3$ depends on the rate constraints of that user. Also, $\alpha_i>0$ if user $i$ performs block Markov coding. These results provide the practical guidance for the optimal implementation of our scheme based on the link states and are discussed in Section \ref{sec:dis}.

\subsection{Required relay transmit power}
As a consequence of the link regime analysis, both users must always use full power but the relay does not necessarily need to use full power. The relay power depends on the transmission scheme as shown in Lemma 1 and 2.
\begin{lem}\label{lemmaUsersFullPwr}
Given any set of link states, both users will always use full power. The relay will use full power if either user performs block Markov coding.
\end{lem}
\begin{proof}\label{proofUsersFullPwr} 
Omitted due to space.
%
\end{proof}
The proof to Lemma 1 is straightforward and follows from the idea that when a user performs block Markov coding, it is always beneficial to increase the power of the block Markov component at both the user and relay.

\begin{lem}\label{pwrSave}
When the relay performs only independent coding for both users (regimes ($\mathcal{R}2, \mathcal{T}3$) and ($\mathcal{R}2, \mathcal{T}4$)), the optimal relay power is 
\begin{align}\label{beta3}
\beta_3=\max\left(\frac{g_{r2}^2-g_{12}^2(1+g_{r1}^2P)}{g_{1r}^2(g_{r1}^2P+1)}P,\frac{g_{r1}^2-g_{21}^2}{g_{2r}^2}P\right).
\end{align}
\end{lem}
\begin{proof} See Appendix A.
\end{proof}
This power savings is a novel result that follows naturally from the link-state perspective. When the relay performs independent coding in ($\mathcal{R}1, \mathcal{T}2$), ($\mathcal{R}1, \mathcal{T}3$), ($\mathcal{R}2, \mathcal{T}1$) or ($\mathcal{R}2, \mathcal{T}2$), similar results apply but with a different $\beta_3$ omitted here due to space.

\section{Discussion and Numerical Results}\label{sec:dis}
\subsection{Discussion of link regimes}
Table 1 describes the optimal relay techniques for each link-state regime when user 1's rate is prioritized over that of user 2. It includes special cases of the classical one-way relay channel, in which only one user uses the relay and the other user performs direct transmission. The TWRC occurs when both user-to-relay links, $g_{r1}^2$ and $g_{r2}^2$, are stronger than their respective direct links (i.e. outside of $\mathcal{R}1$, $\mathcal{T}1$, and $\mathcal{T}2$), in which the relay actively forwards information for both users. When there are no direct links between the users, such as when the users are too far apart to communicate, $g_{12}^2=g_{21}^2=0$ in Def. \ref{rateregiondef} and it is optimal to use the TWRC.

\begin{figure}[t] \label{fig4}
    \begin{center}
		\includegraphics[width=0.4\textwidth]{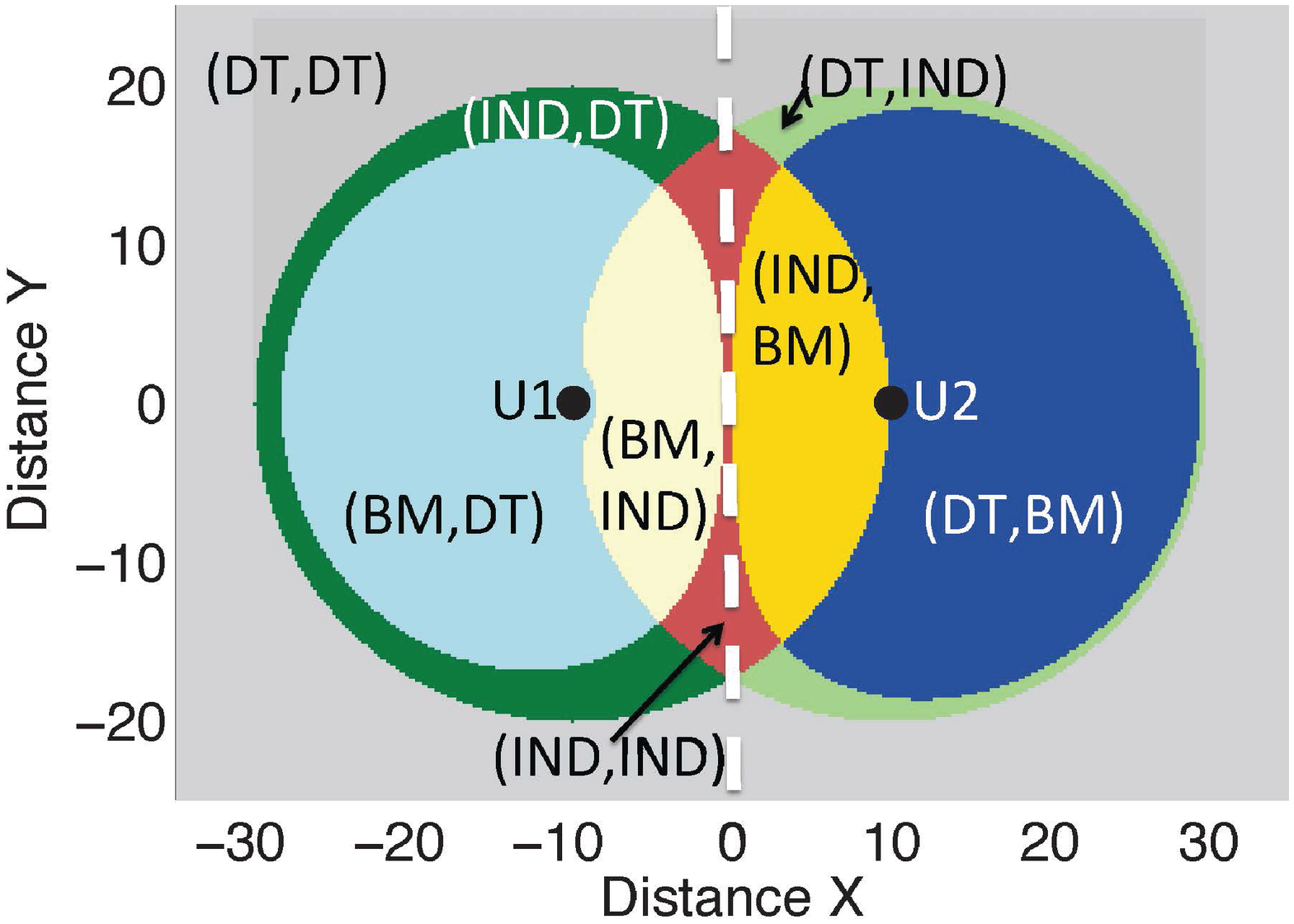}
	
    \caption{Optimal transmission as a function of distance for FDD system}
    \label{fig:distSim}
    \end{center}
\vspace*{-7mm}
\end{figure}

Results shown in Table 1 demonstrate that when the user-to-relay link is just marginally stronger than the direct links, independent coding is employed to provide a network coding gain. With only independent coding, the source uses full power for the new message in each block and the relay uses just enough power to forward the old message. It is interesting to note that in this case, the relay does not need to use full power, as shown in Lemma 2. As the user-to-relay link becomes stronger, block Markov coding is also used. Since the link is strong enough, the gain from coherent beamforming based on block Markov coding outplays the loss from splitting power at the sources between old and new messages. In this case, the relay uses full power to forward the messages.

Note that Table 1 is asymmetric with respect to user 1 and user 2. This is because the table is optimized for $\mu \in (1/2,1]$, corresponding to the lower corner of the rate region, which favors the rate of user 1 over that of user 2. From this chart, we see that as the user-to-relay link sufficiently improves, it is optimal to allocate some power to the block Markov component of our scheme. When both $g_{r1}^2$ and $g_{r2}^2$ are strong, specifically in ($\mathcal{R}3$, $\mathcal{T}4$) and ($\mathcal{R}3$, $\mathcal{T}5$), both users employ both block Markov coding and independent coding to achieve the largest rate region. Since Table 1 is derived for $\mu \in (1/2,1]$ there are more cases in which user 1 performs block Markov coding. These regimes are transposed for $\mu \in [0,1/2)$.

This theme of transitioning from direct transmission to independent coding to block Markov coding as the link state improves is illustrated in Fig. \ref{fig:distSim}. The distance between the two users is fixed at 20 meters while the relay location varies over the entire X-Y plane, typical inter-node distances for a small cell in 5G systems. A frequency-division duplex (FDD) system is assumed with $P=1$, a path loss exponent ($\gamma_1$) of 2.3 for $g_{r1},g_{2r}$ and $g_{21}$, and $\gamma_2=3.6$ for $g_{1r},g_{r2},$ and $g_{12}$, where $g_{ij}=1/d_{ij}^{\gamma/2}$ and $d$ is the distance (meters) between nodes $i$ and $j$. Because this is a FDD system and user 1 is prioritized, the regions are not symmetric. When the relay is between the two users in Fig. \ref{fig:distSim}, the relay performs independent coding for both users. The trend of moving from independent coding to block Markov coding as the link improves is exemplified by the TWRC gold and yellow regions. These regions in which one user performs independent coding and the other performs block Markov coding also highlight the need for our composite scheme. By knowing which scheme to employ for which user, the achievable rate region is considerably improved.
\begin{figure}[t] \label{fig3}
    \begin{center}

		\includegraphics[width=0.44\textwidth]{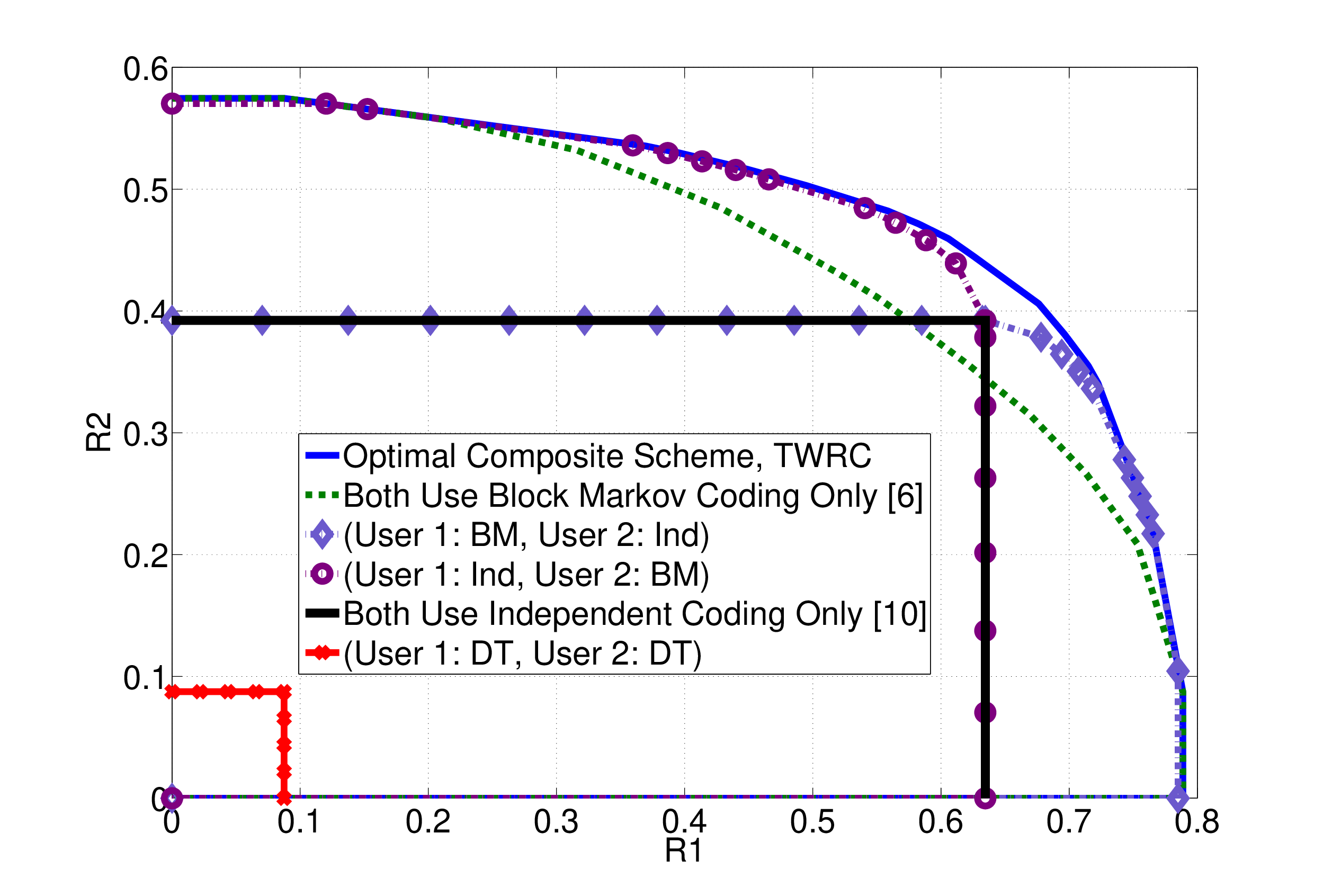}
	
    \caption{Rate region for link gains in ($\mathcal{R}3,\mathcal{T}5$) with $g_{21}=0.25,\ g_{r1}=1,\ g_{12}=0.25,\ g_{r2}=1,\ g_{1r}=0.5,\ g_{2r}=0.7$}
    \label{fig:simulation}
    \end{center}
\vspace*{-7mm}
\end{figure}

\subsection{Numerical results}
In this section, we present numerical results to verify the analysis in Theorem \ref{tablethm}. The simulation plots the rate region in Theorem 1 by varying each power parameter subject to the power constraints in (\ref{powerconstraints}) by step sizes of 0.05 with power $P=1$ and link states as in the caption of Fig. 3. At each step, $J_1,J_2,J_3,J_4,$ and $J_5$ are calculated and the resulting (R1,R2) point satisfying (\ref{RateConstraints}) is obtained, where R1, R2 are the rates of user 1 and 2, respectively. The convex hull of all computed points (R1,R2) is taken to obtain the rate region.

%

Fig. \ref{fig:simulation} shows a nontrivial case in which both users perform both block Markov and independent coding in ($\mathcal{R}3,\mathcal{T}5$). Based on numerical results in Fig. \ref{fig:simulation}, our composite scheme significantly outperforms each individual coding scheme. Fig \ref{fig:simulation} also illustrates how combining block Markov coding and independent coding outperforms time sharing between user 1 performing only block Markov coding with user 2 performing only independent coding and vice versa. Thus, to obtain the largest rate region, both components are necessary.

Also evident from Fig. \ref{fig:simulation} is how drastically the rate region improves when utilizing the relay. The achievable rate region using direct transmission is remarkably smaller than any of the relaying techniques explored in this paper, thus demonstrating the need for intelligent relaying in 5G networks. 

Furthermore, as demonstrated in Fig. \ref{fig:rPower}, the relay power savings in Lemma 2 are substantial in the regions in which the relay performs only independent coding. In Fig. \ref{fig:rPower}, where the relay is assumed to be on the white dotted line between the two users in Fig. \ref{fig:distSim}, only a fraction of the full relay power is necessary. This energy efficiency is a critical advantage when considering utilizing idle nodes as relays in 5G networks.

\begin{figure}[t] \label{fig5}
    \begin{center}

		\includegraphics[width=0.33\textwidth]{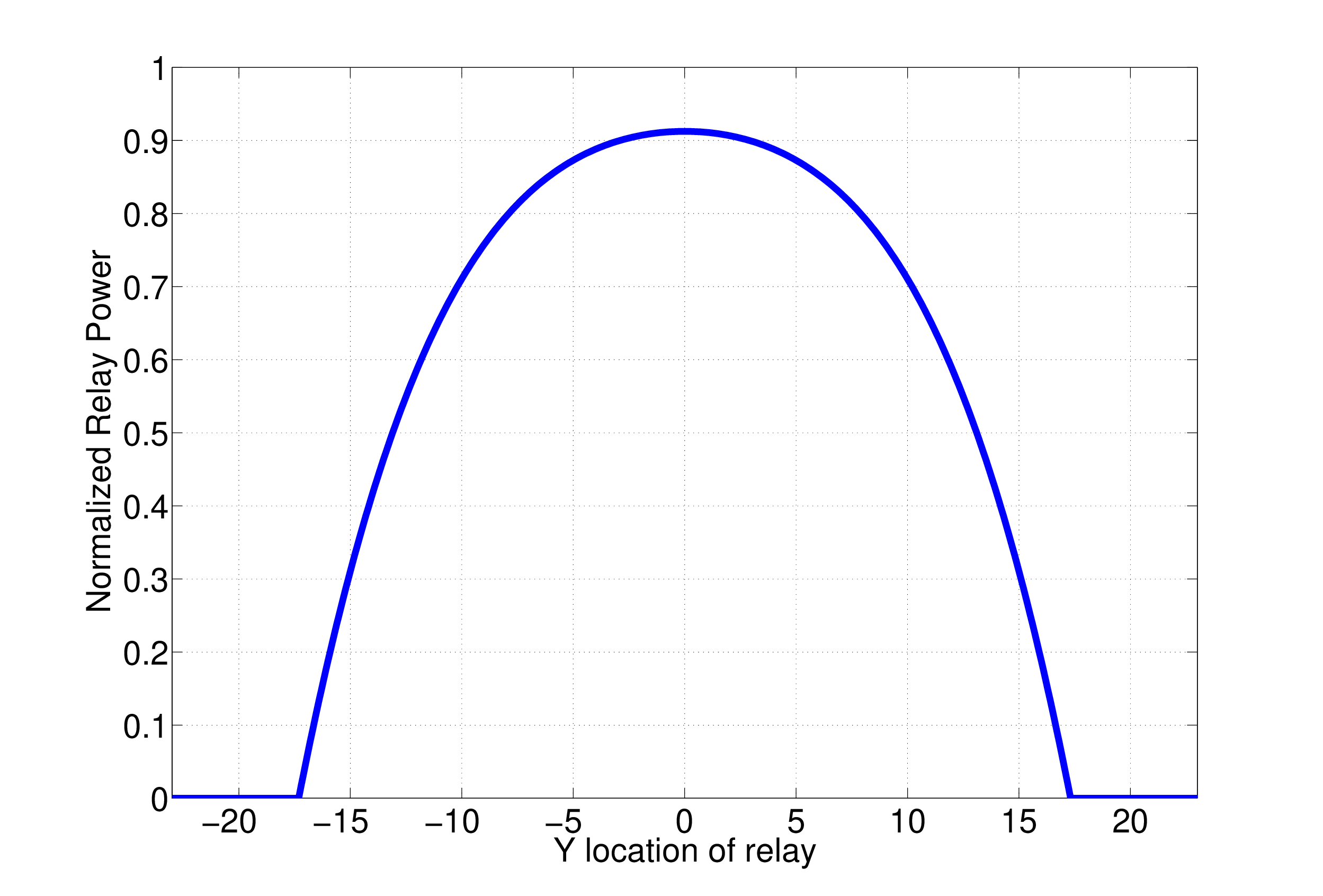}
	
    \caption{Relay power required to implement the composite DF scheme when relay is located on white dotted line in Fig. 2}
    \label{fig:rPower}
    \end{center}
\vspace*{-7mm}
\end{figure}

%

\section{Conclusion}\label{sec:conclusion}
By combining independent coding and block Markov coding in a composite TWRC transmission scheme, we achieve a larger rate region than that of either individual technique or timesharing between the two. When independent coding is optimal, the relay does not need to use full power, further motivation for an idle user to employ this relaying strategy in 5G. Both components of the composite DF scheme are necessary, and by combining them optimally we achieve a significant improvement in the rates of both users. We analytically derived the link-state regimes for the optimality of different transmission schemes, an important step towards practical implementation of two-way DF relaying in systems such as 5G cellular.
\section*{Acknowledgment}
This work has been supported in part by the Office of Naval Research (ONR, Grant N00014-14-1-0645) and National Science Foundation Graduate Research Fellowship Program (NSF, Grant No. DGE-1325256). Any opinions, findings, and conclusions or recommendations expressed in this material are those of the authors and do not necessarily reflect the views of the Office of Naval Research or NSF. Lisa Pinals acknowledges the IEEE Life Member Graduate Study Fellowship.

\appendices
\section{}\label{appendixA}
\subsection{Proof of Theorem 2}
Due to limited space, we will restrict our attention to an interesting representative case of Theorem \ref{tablethm} for the TWRC ($g_{r1}^2>g_{21}^2$ and $g_{r2}^2>g_{12}^2(1+g_{r1}^2P)$) in which user 1 performs only independent coding, while user 2 performs only block Markov coding in link-state regime ($\mathcal{R}2, \mathcal{T}5$). The full analysis will be shown in an upcoming journal article.

From optimization problem (\ref{optimization}), form the Lagrangian as:
\begin{align}\label{LagrangianGen}
\mathcal{L} \triangleq &-\mu R_1 - (1-\mu) R_2 - \lambda_1(J_1-R_1) -\lambda_2(J_2-R_1)\nonumber\\ 
&-\lambda_3(J_3-R_2)- \lambda_4(J_4-R_2) - \lambda_5(J_5-R_1-R_2) \nonumber\\
&-\lambda_6(P-\alpha_1-\beta_1) - \lambda_7(P-\alpha_2-\beta_2)\nonumber\\
&-\lambda_8(P-k_1\alpha_1-k_2\alpha_2-\beta_3). 
\end{align}
where $\lambda_i\geq0$ are the Lagrange dual variables.
The KKT conditions can then be written as
\begin{align}\label{kktGen}
  & \nabla\mathcal{L}=\lambda_1^\star(J_1-R_1^\star)\! =\!  \lambda_2^\star(J_2-R_1^\star)\! = \!
  \lambda_3^\star(J_3-R_2^\star)\! = \! 0\nonumber\\
   &\lambda_4^\star(J_4-R_2^\star)\! = \! \lambda_5^\star(J_5-R_1^\star-R_2^\star)\! =\! \lambda_6^\star(P-\alpha_1^\star-\beta_1^\star) \! =\! 0\nonumber\\
  &\lambda_7^\star(P-\alpha_2^\star-\beta_2^\star)\! =\! 
  \lambda_8^\star(P-k_1\alpha_1^\star-k_2\alpha_2^\star-\beta_3^\star)\! =\! 0
\end{align}
\noindent where all the primal and dual variables are non-negative.

Here we examine the case $\lambda_1>0; \ \lambda_2,\lambda_3=0; \ \lambda_4, \lambda_5>0$. As such, $R_1^*=J_1$ and $R_2^*=J_4=J_5-J_1$. Using these values for the dual variables, the derivative of the Lagrangian with respect to $\alpha_1$ is:
\begin{align}\label{case1ceq_difshort}\
  \nabla_{\alpha_1}\mathcal{L} &= \lambda_6 + \lambda_8 k_1
\end{align}
\noindent Since user 1 always uses full power (Lemma 1), $\lambda_6>0$. From $\nabla_{\alpha_1}\mathcal{L}$, we immediately see that $\alpha_1^*=0$ and thus user 1 performs independent coding only.
Next we can write the derivatives of the Lagrangian with respect to $k_2$ and $\beta_3$ as:
\begin{align}\label{k2beta3}
  &\nabla_{k_2}\mathcal{L} =\alpha_2 \! \left(\! \lambda_8 \! - \! \lambda_4 \!
  \frac{g_{12}g_{1r}k_2^{-0.5}\! +\! g_{1r}^2}
  {n_1}\!\right), \nabla_{\beta_3}\mathcal{L} =\lambda_8 \! - \!\lambda_4\frac{g_{1r}^2}{n_1}\nonumber\\
    &\text{where }n_1\! = \! g_{12}^2P\! +\! 2g_{12}g_{1r}\sqrt{k_2}\alpha_2\! + \! g_{1r}^2(k_2\alpha_2\! +\! \beta_3)\! +\! 1
\end{align}
Making $\nabla_{k_2}=0$ produces a $\nabla_{\beta_3}\mathcal{L}$ that is always positive. Therefore we must minimize $\beta_3$.

From the rate constraints in this case, $R_1^*=J_1\le J_2$. Since user 1 performs independent coding ($\alpha^*_1=0, \beta^*_1=P$), using (\ref{RateConstraints2}) we can write this inequality as:
\begin{align}\label{gr1bounds}
g_{r1}^2&\le g_{21}^2+g_{2r}^2\frac{\beta_3}{P}, \text{ meaning }\beta_3\ge\frac{g_{r1}^2-g_{21}^2}{g_{2r}^2}P
\end{align}
However, from the KKT conditions, $\beta_3$ must be minimized. Therefore (\ref{gr1bounds}) is satisfied with equality for $\beta_3$. Note that we obtained the value of $\beta_3$ using only rate constraints for user 1. Therefore, user 2 only performs block Markov coding. Since $\beta_3 \in[0,P]$, from (\ref{gr1bounds}) $g_{r1}^2$ is bounded from above as: 
\begin{align}\label{gr1compbounds}
g_{r1}^2\le g_{21}^2+g_{2r}^2
\end{align}
Finally, because $\lambda_4\!>\!0$ and $\lambda_5\!>\!0$ we must have $J_4\!=\!J_5-J_1$. Simplify noting that $k_2\alpha_2+\beta_3=P$ and $\alpha_1=0$:
\begin{align}\label{j4eqj5j1}
 g_{r2}^2\!&=\!(g_{12}^2\!\frac{P}{\beta_2}\!+\!\frac{2 g_{12}g_{1r} \sqrt{k_2} (P\!-\!\beta_2)}{\beta_2}\!+\!\frac{g_{1r}^2 P}{\beta_2})(1\!+\!g_{r1}^2\!P)
\end{align}
Since this equation is satisfied with equality, we know  $g_{r2}^2$ is greater than the minimum of the right hand side and less than the maximum of (\ref{j4eqj5j1}) over the range of $\beta_2$. With $\beta_2\rightarrow 0$, the maximum is infinity, and hence always satisfied. The minimum occurs when $\beta_2=P$ and we obtain a lower bound for $g_{r2}^2$:
\begin{align}\label{bounds1cgr2}
g_{r2}^2 &\ge (g_{12}^2 + g_{1r}^2)(1+g_{r1}^2P)
\end{align}
Since we know the value for $\beta_3$ from (\ref{gr1bounds}), we can write (\ref{j4eqj5j1}) in terms of $\alpha_2$ and $k_2$ using $\alpha_2+ \beta_2\! =\! P$. Then there are two unknowns ($\alpha_2, k_2$) and two equations: $J_4\! =\! J_5\! -\! J_1$ and $k_2\alpha_2\! +\! \beta_3\! =\! P$ to solve for the missing power allocation parameters. 

\subsection{Proof of Lemma 2}
Examine the case with dual variables $\lambda_1>0$; $\lambda_2,\lambda_3,\lambda_4\!=\!0; \ \lambda_5\!>\!0$ in (\ref{optimization}). The derivative of the Lagrangian (\ref{LagrangianGen}) with respect to $\alpha_1$ is (\ref{case1ceq_difshort}) and user 1 performs independent coding for the same reason as the case considered in Appendix A. The bound for $g_{r1}^2$ in (\ref{gr1compbounds}) applies, as does $\beta_3$ in (\ref{gr1bounds}). The derivative of the Lagrangian with respect to $\alpha_2$ is 
\begin{align}
\nabla_{\alpha_2}\mathcal{L} =\lambda_7 + \lambda_8 k_2
\end{align}
and we immediately see that user 2 also performs independent coding. Since $J_5-J_1 \leq J_4$,
\begin{align}\label{case1a}
g_{r2}^2 &\le (g_{12}^2  + g_{1r}^2 \frac{\beta_3}{P})(1+g_{r1}^2P)
\end{align}
Solving (\ref{case1a}) yields another expression for $\beta_3$, thus $\beta_3$ must satisfy both (\ref{case1a}) and (\ref{gr1bounds}), leading to the optimal $\beta_3$ in (\ref{beta3}). Thus, since $\alpha_1=\alpha_2=0$, the relay's power is equal to $\beta_3$, and in general is less than P when both users perform independent coding in link regimes  ($\mathcal{R}2, \mathcal{T}3$) and  ($\mathcal{R}2, \mathcal{T}4$).

\bibliographystyle{IEEEtran}
\bibliography{references}
\end{document}